\documentclass{scrartcl}
\usepackage{url}
\usepackage{textcomp}
\usepackage{pifont}
\usepackage{latexsym} 
\usepackage{amsfonts} 
\usepackage{amssymb}
\usepackage{amsmath}
\usepackage{amsthm}
\usepackage{stmaryrd}
\usepackage[english]{babel}
\usepackage{pgf}
\usepackage{stmaryrd}
\usepackage{acronym}
\usepackage{slashed} 
\usepackage{tikz}
\usepackage{mathtools}
\usepackage{color}
\usepackage{pgf}
\usepackage{tikz}
\usetikzlibrary{backgrounds}
\usetikzlibrary{fit}
\usetikzlibrary{matrix}
\usepackage{xargs}
\usepackage{epo}

\newcommand{\email}[1]{\url{#1}}

\newtheorem{theorem}{Theorem}[section]

\newtheorem{lemma}[theorem]{Lemma}
\newtheorem{claim}[theorem]{Claim}
\newtheorem{proposition}[theorem]{Proposition}
\newtheorem{definition}[theorem]{Definition}

\newtheorem{example}[theorem]{Example}

\begin{document}

\title{A Path Order for Rewrite Systems that Compute Exponential Time Functions%
\thanks{This research is supported by FWF (Austrian Science Fund) projects P20133-N15.}%
\\\small{Technical Report}}

\author{Martin Avanzini\\
Institute of Computer Science\\
University of Innsbruck, Austria\\
\url{martin.avanzini@uibk.ac.at}
\and
Naohi Eguchi\\
School of Information Science\\
Japan Advanced Institute of Science and Technology, Japan\\
\url{n-eguchi@jaist.ac.jp}
\and 
Georg Moser\\
Institute of Computer Science\\
University of Innsbruck, Austria\\
\url{georg.moser@uibk.ac.at}}
\maketitle

\begin{abstract}
In this paper we  present a  new  path order for 
rewrite systems,  the \emph{exponential path order} \EPOSTAR.  
Suppose a term rewrite system $\RS$ is  compatible  with \EPOSTAR,  
then  the  runtime  complexity of $\RS$ is bounded from  above by an exponential  function. 
Further, the  class  of  function  computed  by  a  rewrite  system
compatible with \EPOSTAR\ equals  the class of functions computable in
exponential time on  a Turing machine. 
\end{abstract}

\newpage
\tableofcontents
\newpage

\section{Introduction}\label{s:intro}

In this paper we are concerned with the complexity analysis of term rewrite
systems (TRSs) and the ramifications of such an analysis in implicit computational
complexity (ICC for short).

Term rewriting is a conceptually simple but powerful abstract model of
computation that underlies much of declarative programming. 
In rewriting, proving termination is an important research field. 
Powerful methods have been introduced to establish 
termination of TRSs (see~\cite{BN98} for an overview).
In order to assess the complexity of a (terminating) TRS
it is natural to look at the maximal length of derivations. 
More precisely in~\cite{HL89} the \emph{derivational complexity} of a TRS
is studied, where the derivational complexity function relates the 
length of a longest derivation sequence to the size of the initial term.
A more fine-grained approach is introduced in~\cite{CKS89} (compare also~\cite{HM08}),
where the derivational complexity function is refined so that in principle
only argument normalised (aka basic) terms are considered. In the following
we refer to the latter notion as the \emph{runtime complexity} of a TRS. 

In recent years the field of complexity analysis of rewrite systems matured
and some advances towards an automated complexity analysis of TRSs evolved
(see~\cite{M09} for an overview). The current focus of modern complexity
analysis of rewrite systems is on techniques that yield \emph{polynomial} runtime
complexity. In this paper we study a complementary view and introduce
the path order \EPOSTAR. 
The definition of \EPOSTAR\ makes use of \emph{tiering}~\cite{BC92} and 
is strongly influenced by a very recent term-rewriting characterisation of
the class of functions computable in exponential time by Arai and the second author~\cite{AraiE09}.
\begin{example}
\label{ex:1}
Consider the following TRS $\RSfib$ which is easily seen to represent
the computation of the $n^{\text{th}}$ Fibonacci number. 
\begin{align*}
    \ffib(x) & \to \dfib(x,0) 
    & \dfib(0,y) & \to \ms(y)  \\
    \dfib(\ms(0),y) & \to \ms(y) 
    & \dfib(\ms(\ms(x)),y) & \to \dfib(\ms(x),\dfib(x,y))
  \end{align*}
\end{example}

Then all rules in the TRS $\RSfib$ can be oriented with \EPOSTAR, which
allows us to (automatically) deduce that the runtime complexity of this
system is exponential. 
Exploiting graph rewriting we show that any TRS compatible
with $\EPOSTAR$ is computable in exponential time on a Turing machine.
Conversely we show that any function $f$ that can be
computed in exponential time can be computed by a TRS $\RS(f)$ such that
$\RS(f)$ is compatible with \EPOSTAR. Hence we provide soundness and
completeness for \EPOSTAR\ with respect to the class of functions computable
in exponential time.

\medskip
\emph{Related Work.}
With respect to rewriting we mention~\cite{EWZ08}, where it is shown that
\emph{matrix interpretations} yield exponential derivational complexity,
hence at most exponential runtime complexity. Our work is also directly
related to work in ICC (see~\cite{BMR09} for an overview). 
Here we want to mention~\cite{BCMT01,BM10} were alternative
characterisations of the class of functions computable in exponential time are given.
For less directly related work we cite~\cite{AJK08}, where a complete characterisation
of (imperative) programs that admit linear and polynomial runtime complexity is
established. As these characterisations are decidable, we obtain a decision
procedure for programs that admit a runtime complexity that is at least exponential.

\medskip
The remaining of the paper is organised as follows. 
In Section \ref{s:prel} we recall definitions. 
In Section \ref{s:epo} we introduce the intermediate order \EPO. 
Our main result is presented in Section \ref{s:epostar}.
In Section \ref{s:impl} we show how the ordering constraints imposed by \EPOSTAR~can
be expressed in propositional logic. Using a state-of-the-art SAT-solvers, 
this gives us a machinery to automatically verify compatibility of 
TRSs with \EPOSTAR.
Finally, we conclude in Section \ref{s:conclusion}.


\section{Preliminaries}\label{s:prel}

We briefly recall central definitions and introduce employed notions.
We assume a basic understanding of complexity theory~\cite{Kozen}.
We write $\NAT$ for the set of \emph{natural numbers}.
Let $R \subseteq {A \times A}$ be a binary relation. 
We write $a \mathrel{R} b$ instead of $(a,b) \in R$.
We denote by $R^+$ the transitive and by $R^*$ the transitive and reflexive closure of $R$.
Further, $R^n$ denotes the $n$-fold composition of $R$.
The relation $R$ is \emph{well-founded} if there exists no infinite 
sequence $a_1 \mathrel{R} a_2 \mathrel{R} \dots$, 
the relation $R$ is \emph{finitely branching} if $\{b \mid a \mathrel{R} b\}$
is finite for all $a \in A$.
A \emph{preorder} is a reflexive and transitive binary relation.
If $\qp$ is a preorder, we write ${\ep} \defi {\qp} \cap {\preccurlyeq}$ and ${\sp} \defi {\qp} \setminus {\ep}$ do denote
the \emph{equivalence} and \emph{strict part} of $\qp$ respectively. 

We follow the notions of \emph{term rewriting} from \cite{BN98}.
Let $\VS$ denote a countably infinite set of variables and $\FS$ a 
signature, i.e, a set of function symbols with associated arities. 
With $\ar(f) \in \NAT$ we denote the \emph{arity} of $f$.
The set of terms over $\FS$ and $\VS$ is denoted by $\TERMS$. 
We denote by $\vec{s}, \vec{t}, \dots$ sequences of terms, 
and for a set of terms $T$ we write $\vec{t} \subseteq T$ to 
indicate that for each $t_i$ appearing in $\vec{t}$, $t_i \in T$. 
We suppose that the signature $\FS$ is partitioned into \emph{defined symbols} $\DS$
and \emph{constructors} $\CS$.
The set of \emph{basic} terms $\Tb \subseteq \TERMS$ is 
defined as $\Tb \defi \{ f(\seq{t}) \mid f \in \DS \text{ and $t_i \in \TA(\CS,\VS)$ for $i \in \{1,\dots,n\}$}\}$.

We write $\subterm$ and $\superterm$ to denote the \emph{subterm} and respectively 
\emph{superterm} relation, the strict part 
of $\subterm$ (respectively $\superterm$) is denoted by $\subtermstrict$ (respectively $\supertermstrict$).
Let $t$ be a term.
We denote by $\size{t}$ and $\depth(t)$ the the \emph{size} and \emph{depth} of the term $t$.
If $t = f(\seq{t})$, we denote by $\rt(t)$ the \emph{root symbol} $f$.
Let $\hole$ be a constant not appearing in $\FS$. 
Elements from $\TA(\FS \cup \{\hole\},\VS)$ with exactly one occurrence of $\hole$ 
are called \emph{contexts} and denoted by $C$, 
$C[t]$ denotes the term obtained by replacing $\hole$ in $C$ by $t$.
A \emph{substitution} is a mapping $\ofdom{\sigma}{\VS \to \TA(\FS,\VS)}$, 
extended to terms in the obvious way.
We write $t\sigma$ instead of $\sigma(t)$.
 A \emph{quasi-precedence} (or simply \emph{precedence}) is a preorder ${\qp} = {{\sp} \uplus {\ep}}$ 
on the signature $\FS$ so that the strict part $\sp$ is well-founded.

A \emph{term rewrite system} (\emph{TRS} for short) is
a set of rewrite rules $l \to r$ such that $l \not \in \VS$ and all variables 
in $r$ occur in $l$.
We always use $\RS$ to denote a TRS. If not mentioned otherwise, $\RS$ is \emph{finite}.
We denote by $\rew$ the \emph{rewrite relation} as induced by $\RS$, 
i.e., $s \rew t$ if $s = C[l\sigma]$ and $t = C[r\sigma]$ for some rule ${l \to r} \in \RS$.
With $\irew$ we denote the \emph{innermost rewrite relation}, that is,
the restriction of $\rew$ where additionally all proper subterms of $l\sigma$ are normal forms.
Here a term $t$ is in \emph{normal form} if there exists no $u$ such that $t \rss u$.
The set of all \emph{normal forms} of $\RS$ is denoted by $\NF(\RS)$.
We write $t \rsn u$ (respectively $t \irsn u$) if $t \rss u$ (respectively $t \irss u$)
and $u \in \NF(\RS)$.
A rewrite step is a \emph{root step} if $C=\hole$ in the definition of $\rew$.
The TRS $\RS$ is a \emph{constructor} TRS if left-hand sides are basic terms, 
$\RS$ is \emph{completely defined} if each defined symbol is completely defined.
Here a symbol is completely defined if it does not occurring in any normal form.
The TRS $\RS$ is called \emph{terminating} if $\rew$ is well-founded, 
$\RS$ is \emph{confluent} if for all terms $s, t_1, t_2$ 
with $s \rss t_1$ and $s \rss t_2$,
there exists $u$ such that $t_1 \rss u$ and $t_2 \rss u$.

Let $\to$ be a finitely branching, well-founded binary relation on terms. 
The \emph{derivation height} of a term $t$ with respect to $\to$ 
is given by $\dl(t,\to) \defi \max\{ n \mid \exists u.~t \to^n u \}$. 
The (innermost) \emph{runtime complexity} of the TRS $\RS$ is defined 
as 
$$\rc^{\text{\scriptsize{(\textsf{i})}}}_{\RS}(n) 
\defi \max \{\dl(t,\to) \mid t \in \Tb \text{ and } \size{t} \leqslant n\} \tkom 
$$ 
where $\to$ denotes $\rew$ or $\irew$ respectively.

Let $\M$ be a \emph{Turing machine} (TM for short) \cite{Kozen} 
with alphabet $\Sigma$, and let $w \in \Sigma^*$. 
We say that $\M$ computes $v \in \Sigma^*$ on input $w$,  
if $\M$ accepts $w$, i.e., $\M$ halts in an accepting state,
and $v$ is written on a dedicated output tape.
We say that $\M$ computes a binary relation $R \subseteq {\Sigma^* \times \Sigma^*}$ 
if for all $w,v \in \Sigma^*$ with $w \mathrel{R} v$, $\M$ computes $v$ on input $w$.
Note that if $\M$ is deterministic then $\mathrel{R}$ induces a partial function 
$\ofdom{f_R}{\Sigma^* \to \Sigma^*}$, 
we also say that $\M$ computes the function $f_R$.

Let $\ofdom{S}{\NAT \to \NAT}$ denote a bounding function.
We say that $\M$ runs in time $S(n)$ if for all but finitely many inputs $w \in \Sigma^*$, 
no computation is longer than $S(\size{w})$. Here $\size{w}$ refers to the length of 
the input $w$.
We denote by $\FTIME(S(n))$ the class of functions computable 
by some TM $\M$ in time $S(n)$. 
Then $\FP \defi \FTIME(\bigO(n^k))$ where $k\in\NAT$
is the class of \emph{polynomial-time computable functions}.
Of particular interest for this paper is the 
class of \emph{exponential-time computable functions} 
$\FEXP \defi \FTIME(2^{\bigO(n^k)})$ where $k \in \NAT$.


\section{Exponential Path Order \EPO}\label{s:epo}

In this section, we introduce an intermediate order EPO, extending the
definitions and results originally presented in \cite{Eguchi09}.
The path order EPO is defined over \emph{sequences} of terms from $\TA(\FS,\VS)$.
To denote sequences, we use an auxiliary function symbol $\List$.
The function symbol $\List$ is variadic, i.e.,
the arity of $\List$ is finite, but arbitrary.
We write $\lseq{t}$ instead of $\List(\seq{t})$.
For sequences $\lseq{s}$ and $\lseq[m]{t}$, we write 
$\lseq{s} \append \lseq[m]{t}$ to denote the concatenation
$\lst{s_1~\cdots~s_n~t_1~\cdots~t_m}$.
We write $\TAL(\FS,\VS)$ for the set of finite sequences of terms from $\TA(\FS,\VS)$, 
i.e. $\TAL(\FS,\VS) \defi \{ \lseq{s} \mid n \in \NAT \text{ and } \seq{s} \in \TA(\FS,\VS)  \}$.
Each term 
$t \in \TA(\FS,\VS)$ is identified with the single list
$\lst{t} = \List(t) \in \TAL(\FS,\VS)$.
This identification allows us to ensure
$\TA(\FS,\VS) \subseteq \TAL(\FS,\VS)$.
We use $a, b, c, \dots $ to denote elements of
$\TAL(\FS,\VS)$, possibly extending them by subscripts.

Let ${\qp}$ to denote a (quasi-)precedence on the signature $\FS$.
We lift the equivalence ${\ep} \subseteq {\qp}$ on $\FS$ to terms in the obvious way:
$s \eqi t$ iff 
(i) $s = t$, or
(ii) $s = f(\seq{s})$, $t = g(\seq{t})$, $f \ep g$
 and $s_i \ep t_i$ for all $i \in \{1,\dots,n\}$.
Further, we write $\esupertermstrict$ for the \emph{superterm relation modulo
term equivalence $\eqi$}, defined by 
$f(\seq{s}) \esupertermstrict t$ if $s_i \esuperterm t$ 
for some $i \in \{ 1, \dots, m \}$. 
Here ${\esuperterm} \defi {\esupertermstrict} \cup {\ep}$.
The precedence $\qp$ induces a \emph{rank} $\rk(f) \in \NAT$ on $f \in \FS$ as follows:
$\rk(f) = \max \{1 + \rk(g) \mid g \in \FS \text{ and } f \sp g\}$, 
where we suppose $\max \varnothing = 0$. 

\begin{definition}\label{d:epo}
  Let $a, b \in \TAL(\FS,\VS)$, and let $k \geqslant 1$.
  Below we assume $f,g \in \FS$.
  We define $a \epo[k] b$ with respect to the precedence $\qp$ if either
  \begin{enumerate}
  \item \label{d:epo:1} 
    $a = f(\seq[m]{s})$
    and
    $s_i \epoeq[k] b$ for some $i \in \{1, \dots, m \}$, or
  \item \label{d:epo:2} 
    $a = f(\seq[m]{s})$, 
    $b = \lseq[n]{t}$ with $n = 0$ or $2 \leqslant n \leqslant k$,
    $f$ is a defined function symbol,
    and $a \epo[k] t_j$ for all $j \in \{1, \dots, n\}$, or
  \item \label{d:epo:3} 
    $a = f(\seq[m]{s})$, 
    $b = g(\seq[n]{t})$ with $n \leqslant k$, 
    $f$ is a defined function symbol with $f \sp g$,
    and $a$ is a strict superterm (modulo $\eqi$) of all $t_j$ ($j \in \{1, \dots, n\}$), or
  \item \label{d:epo:4} 
    $a = \lseq[m]{s}$, 
    $b = b_1 \append \cdots \append b_m$, 
    and for some $j \in \{1, \dots, m\}$, 
    \begin{itemize}
    \item $s_1 \eqi b_1$, \dots, $s_{j-1} \eqi b_{j-1}$,
    \item $s_j \epo[k] b_j$, and
    \item $s_{j+1} \epoeq[k] b_{j+1}$, \dots, $s_{m} \epoeq[k] b_{m}$, or
    \end{itemize}
  \item \label{d:epo:5} 
    $a = f(\seq[m]{s})$, 
    $b = g(\seq[n]{t})$ with $n \leqslant k$, 
    $f$ and $g$ are defined function symbols with $f \ep g$, 
    and for some $j \in \{1, \dots, \min(m,n) \}$, 
    \begin{itemize}
    \item $s_1 \eqi t_1$, \dots, $s_{j-1} \eqi t_{j-1}$, 
    \item $s_j \esupertermstrict t_j$, and
    \item $a \esupertermstrict t_{j+1}$, \dots, $a \esupertermstrict t_n$.
    \end{itemize}
  \end{enumerate}
Here we set ${\epoeq[k]} \defi {\epo[k]} \cup {\ep}$.
Finally, we set ${\epo} \defi \bigcup_{k \geqslant 1} \epo[k]$ and
${\epoeq} \defi \bigcup_{k \geqslant 1} \epoeq[k]$.
\end{definition}

We note that, by Definition \ref{d:epo}.\ref{d:epo:2} with 
$n=0$, we have
$f(\seq[m]{s}) \epo[k] \nil$ for all $k \geqslant 1$ if 
$f$ is a defined function symbol.
It is not difficult to see that
$l \leqslant k$ implies ${\epo[l]} \subseteq {\epo[k]}$.

\begin{lemma}
  \label{permlem}
  Let $a = a_1 \append \cdots \append a_m \in \TAL (\FS,\VS)$ and $j \in \{ 1, \dots, m \}$.
  Suppose that $a_j \epo[k] a_j'$.
  Then $a \epo[k] a_1 \append \cdots \append a_{j - 1} \append a_{j}' \append a_{j + 1} \cdots \append a_m$.
\end{lemma}
\begin{proof}
Put
$a' :=
 a_1 \append \cdots \append a_{j-1} \append
 a_j' \append a_{j+1} \append \cdots \append a_m$.
If 
$a_j \in \TA(\FS,\VS)$, 
then
$a \epo[k] a'$ by Definition \ref{d:epo}.\ref{d:epo:4}.
Hence suppose that
$a_j \not\in \TA(\FS,\VS)$.
Then, there exist $n \geqslant 2$ and
$t_1, \dots, t_n \in \TA(\FS,\VS)$
such that
$a_j = \lseq[n]{t}$. 
Since we have 
$a_j \epo[k] a_j'$, according to Definition \ref{d:epo}.\ref{d:epo:4}
there exist 
$j_0 \in \{ 1, \dots, n \}$ and
$b_1, \dots, b_n \in \TAL(\FS,\VS)$ such that
$a_j' = b_1 \append \cdots \append b_n$,
$t_{j_0} \epo[k] b_{j_0}$, and
$t_i \epoeq[k] b_i$ for every $i \in \{ 1, \dots, n \}$.
Hence, again by Definition \ref{d:epo}.\ref{d:epo:4},
we can conclude
$a \epo[k] a'$.
\end{proof}

Following Arai and Moser \cite{ARM05} we define $\Slow[k]$ that measures the 
$\epo[k]$-descending lengths:
\begin{definition}
We define $\ofdom{\Slow[k]}{\TAL(\FS,\VS) \to \NAT}$ as
$$ 
\Slow[k](a) \defi \max
\{ \Slow[k](b) +1 \mid b \in \TAL(\FS,\VS) \text{ and } a \epo[k] b\} \tpkt 
$$
\end{definition}

\begin{lemma}\label{lemG_k}
  For all $k \geqslant 1$ we have
  (i) ${\esupertermstrict} \subseteq {\epo[k]}$, and
  (ii) if $t \in \TA(\CS,\VS)$ then $\Slow[k](t) = \depth(t)$, and 
  (iii) $ \Slow[k] ( \lseq[m]{t}) = \sum_{i=1}^m \Slow (t_i)$.
\end{lemma}
\begin{proof}
The Properties (i) and (ii) can be shown by straight forward 
inductive arguments.
We prove (iii) for the non-trivial case $m \geqslant 2$.
It is not difficult to check that
$ \Slow[k] ( \lseq[m]{t}) \geqslant \sum_{i=1}^m \Slow (t_i)$.
We show that
$ \Slow[k] ( \lseq[m]{t}) \leqslant \sum_{i=1}^m \Slow[k] (t_i)$
by induction on 
$\Slow[k] ( \lseq[m]{t})$.

Let $a = \lseq[m]{t}$.
Then, it suffices to show that, for any 
$b \in \TAL(\FS,\VS)$,
if $a \epo[k] b$, then
$\Slow[k] (b) < \sum_{i=1}^m \Slow[k] (t_i)$.
Fix 
$b \in \TAL(\FS,\VS)$ and
suppose that $a \epo[k] b$.
Then, by Definition \ref{d:epo}.\ref{d:epo:4}, there exist some
$b_1, \dots, b_m \in \TAL(\FS,\VS)$
and  
$j \in \{ 1, \dots, m \}$ such that
$b = b_1 \append \cdots \append b_m$,
$t_i \epoeq[k] b_i$ for each $i \in \{ 1, \dots m \}$, and
$t_j \epo[k] b_j$.
By the definition of $\Slow[k]$, we have that
$ \Slow[k] (t_i) \geqslant \Slow[k] (b_i)$ for each $i \in \{ 1, \dots m \}$, and
$\Slow[k] (t_j) > \Slow[k] (b_j)$.
Thus 
$$\sum_{i=1}^m \Slow[k] (b_i) < \sum_{i=1}^m \Slow[k] (t_i)$$
 follows.
Let $b_i = [ u_{i,1} \cdots u_{i,n_i}]$
for each $i \in \{ 1, \dots, m \}$.
Then, since $\Slow[k] (b) < \Slow[k] (a)$, 
$\Slow[k] (b) \leqslant \sum_{i=1}^m \sum_{j=1}^{n_i} \Slow[k] (u_{i,j})$
holds by induction hypothesis.
Recalling that for each $i \in \{ 1, \dots, m \}$,
$\sum_{j=1}^{n_i} \Slow[k] (u_{i,j}) \leqslant \Slow[k] (b_i)$
also holds we finally obtain that
\begin{align*}
\Slow[k] (b) 
\leqslant \sum_{i=1}^m \sum_{j=1}^{n_i} \Slow[k] (u_{i,j})
\leqslant \sum_{i=1}^m \Slow[k] (b_i) 
< \sum_{i=1}^m \Slow[k] (t_i) \tpkt
\end{align*}
\end{proof}

\begin{theorem}\label{t:epo}
  Suppose that $f \in \FS$ with arity $n \leqslant k$ and $t_1 , \dots , t_n \in \TA(\FS,\VS)$.
  Let $N \defi \max \{ \Slow[k] ( t_i ) \mid 1 \leqslant i \leqslant n \} +1$.
  Then       
  $$
  \Slow[k] ( f ( t_1, \dots, t_n ) ) 
  \leqslant (k+1)^{N^k \cdot \rk(f) + \sum_{ i=1 }^n N^{ k-i } \Slow[k] ( t_i ) } \tpkt
  $$
\end{theorem}

\begin{proof}
We prove the theorem
by induction on 
$ N^k \cdot \rk (f) + \sum_{ i=1 }^n N^{ k-i } \Slow[k] ( t_i ) $.
Let $t = f(t_1, \dots, t_n)$.
In the base case, $f$ is minimal in the precedence $\sp$ on
the signature $\FS$ and the arguments of $f$ are empty.
Hence, 
$\Slow[k] (t) = \Slow[k] (f) = 1 \leqslant
  (k+1)^{ 
    N^k \cdot \rk (f) + \sum_{ i=1 }^n N^{ k-i } \Slow[k] ( t_i ) 
   }
$.
For the induction case, it suffices to show that, for any
$b \in \TAL(\FS,\VS)$,
if $t \epo[k] b$ then 
$\Slow[k] (b) <
 (k+1)^{ 
    N^k \cdot \rk (f) + \sum_{ i=1 }^n N^{ k-i } \Slow[k] ( t_i ) }$.
The induction case splits into five cases according to the last rule which
concludes $t \epo[k] b$.
We consider the most interesting cases:
\begin{enumerate}
\item \textsc{Case} 
$t_i \epoeq[k] b$ for some
$i \in \{ 1, \dots, n \}$:
In this case,
$$
\Slow[k] (b) \leqslant \Slow[k] (t_i) 
< (k+1)^{N^k \cdot \rk (f) + \sum_{ i=1 }^n N^{ k-i } \Slow[k] ( t_i )}\tpkt
$$
\item \textsc{Case} $b = g(\seq[m]{u})$ where $m \leqslant k$, 
  $g$ is a defined symbol with $f \sp g$ 
  and for all $i \in \{ 1, \dots, m\}$, $t$ is a strict superterm (modulo $\eqi$) of $u_i$:
  Let $ M \defi \max \{ \Slow[k] (u_i) \mid 1 \leqslant i \leqslant m \} + 1 $.
  Then, we have $ M \leqslant N $ since $t$ is a strict superterm (modulo $\eqi$)
  of every $u_i$. 
  We claim
  \begin{align*}
  M^k \cdot \rk (g) + \sum_{ i=1 }^m M^{ k-i } \Slow[k] (u_i) 
  < N^k \cdot \rk (f) + \sum_{ i=1 }^n N^{ k-i } \Slow[k] (t_i) \tpkt
  \end{align*}
  To see this, conceive left- and right-hand side as numbers represented 
  in base $M$ and respectively $N$ 
  of length $k$ (observe $\Slow[k](u_i) < M$ and $\Slow[k](t_i) < N$).
  From $\rk(g) < \rk(f)$ and $ M \leqslant N $ the above inequality is obvious.
  Hence, by induction hypothesis, we conclude

  \begin{align}
    \label{epo:th:eq1}    
          \Slow[k] ( b ) 
    & \leqslant (k+1)^{M^k \cdot \rk (g) + \sum_{ i = 1 }^m M^{ k - i  } \Slow[k] (u_i)} \\
    & < (k+1)^{ N^k \cdot \rk (f) + \sum_{ i=1 }^n N^{ k-i } \Slow[k] (t_i)} \tpkt \nonumber 
  \end{align}

\item \textsc{Case} $b = g(\seq[m]{u})$ where $m \leqslant k$, 
  $g$ is a defined symbol with $f \ep g$ and 
  there exists  $j \in \{ 1, \dots, \min(n,m) \}$ 
  such that $t_i \eqi u_i$ for all $i < j$, 
  $t_j$ is a strict superterm (modulo $\eqi$) of $u_j$, and
  $t$ is a strict superterm (modulo $\eqi$) for all $i > j$:
  Let 
  $ M := \max \{ \Slow[k] (u_i) \mid 1 \leqslant i \leqslant m \} + 1 $
  and consider the following claim:
  \begin{claim}
    $\sum_{ i=1 }^m M^{k-i} \Slow[k] (u_i) < \sum_{i=1}^n N^{k-i} \Slow[k] (t_i)$.
  \end{claim}
  To prove this claim, observe that the assumptions give
  $\Slow[k] (t_i) = \Slow[k] (u_i)$ for all $i < j$,
  $\Slow[k] (t_j) < \Slow[k] (u_j)$, and
  $\Slow[k] (t_i) < \Slow[k] (u_i)$ for all $i < j$:
  This implies that $ M \leqslant N $ and  
  \begin{eqnarray*}
    \sum_{ i=1 }^m M^{ k-i } \Slow[k] (u_i)
    & \leqslant & \sum_{i=1}^{j-1} N^{k-i} \Slow[k] (t_i) + N^{k-j} (\Slow[k] (t_j) - 1)
    + \sum_{ i = j + 1 }^n N^{ k-i } ( N - 1 ) \\
    & < & \sum_{ i=1 }^n N^{ k-i } \Slow[k] (t_i) \tpkt
  \end{eqnarray*}
  The claim together with induction hypothesis yields 
  Equations \eqref{epo:th:eq1} as above, concluding the case.

\item \textsc{Case} 
    $b = \lseq[m]{u}$ where $m = 0$ or $2 \leqslant m \leqslant k$
    and $t \epo[k] u_j$ for all $j \in \{1, \dots, m\}$:
    First suppose $m = 0$, i.e., $b = \nil$. Then $\Slow[k](b) = 0$ 
    by Lemma~\ref{lemG_k} and the Theorem follows trivially.  Hence 
    suppose $2 \leqslant m \leqslant k$.
    From the former cases, it is not difficult to see that
    $$
    \Slow[k] (u_i) 
    \leqslant (k+1)^{ (N^k \cdot \rk (f) + \sum_{i=1}^n N^{k-i} \Slow[k] (t_i)) - 1} \tpkt
    $$
    for all $i \in \{1,\dots,m\}$.
    Therefore by Lemma \ref{lemG_k}, and employing $m \leqslant k$, we see
    \begin{align*}
      \Slow[k] (b) = \sum_{i=1}^{m} \Slow[k] (u_i) 
      & \leqslant k \cdot (k+1)^{(N^k \cdot \rk (f) + \sum_{i=1}^n N^{k-i} \Slow[k] (t_i)) -1} \\
      & < (k+1)^{ N^k \cdot \rk (f) + \sum_{i=1}^n N^{k-i} \Slow[k] (t_i)} \tpkt
    \end{align*}
\end{enumerate}
This completes the proof of the theorem. 
\end{proof}


\section{Exponential Path Order \EPOSTAR}\label{s:epostar}

We now present the \emph{exponential path order} (\EPOSTAR for short), defined 
over terms $\TERMS$.
We call a precedence $\qp$ \emph{admissible} if constructors are minimal, i.e., 
for all defined symbols $f$ we have $f \sp c$ for all constructors $c$.
Throughout the following, we fix ${\qp}$ to 
denote an \emph{admissible} quasi-precedence on $\FS$.
A \emph{safe mapping} $\safe$ on $\FS$ is a function 
$\ofdom{\safe}{\FS \to 2^\NAT}$ that associates 
with every $n$-ary function symbol $f$ the set of \emph{safe argument positions} 
$\{i_1, \dots , i_m\} \subseteq \{1,\dots,n\}$.
Argument positions included in $\safe(f)$ are called \emph{safe},
those not included are 
called \emph{normal} and collected in $\normal(f)$.
For $n$-ary constructors $c$ 
we require that all argument positions are safe, 
i.e., $\safe(c) = \{1,\dots,n\}$.
To simplify the presentation, we write $f(t_{i_1}, \dots, t_{i_k}; t_{j_1}, \dots, t_{j_l})$
for the term $f(\seq{t})$ with $\normal(f) = \{\seq[k]{i}\}$ and $\safe(f) = \{\seq[l]{j}\}$.
We restrict term equivalence $\eqi$ in the definition of $\eqis$ below
so that the separation of arguments through $\safe$ is taken into account:
We define $s \eqis t$ if either (i) $s = t$, or (ii) $s = f(\pseq[l][m]{s})$, $t = g(\pseq[l][m]{t})$ where
$f \ep g$ and $s_i \eqis t_i$ for all $i \in \{1, \dots, m\}$.
The definition of an instance 
$\epostar$ of $\EPOSTAR$ is split into two definitions.
\begin{definition}\label{d:subepo}
  Let $s, t \in \TERMS$ such that $s = f(\pseq[l][m]{s})$.
  Then $s \subepostar t$ if $s_i \subepostareq t$ for some $i \in \{1, \dots, m\}$.
  Further, if $f \in \DS$, then $i \in \normal(f)$.
  Here we set ${\subepostareq} \defi {\subepostar} \cup {\eqis}$.
\end{definition}

\begin{definition}\label{d:epostar}
  Let $s, t \in \TERMS$ such that $s = f(\pseq[l][m]{s})$.
  Then $s \epostar t$ with respect to the admissible precedence $\qp$ and safe mapping $\safe$ if either
  \begin{enumerate}
    \item \label{d:epostar:1} $s_i \epostareq t$ for some $i \in \{1, \dots, l+m\}$, or
    \item \label{d:epostar:2} $t = g(\pseq[k][n]{t})$, $f \sp g$ and
      \begin{enumerate}
      \item $s \subepostar t_1, \dots, s \subepostar t_k$, and 
      \item $s \epostar t_{k+1}, \dots, s \epostar t_{k+n}$, or
      \end{enumerate}
    \item \label{d:epostar:3} $t = g(\pseq[k][n]{t})$, $f \ep g$ and for some $i \in \{1, \dots, \min(l,k)\}$ 
      \begin{enumerate}
      \item $s_1 \eqis t_1, \dots, s_{i-1} \eqis t_{i-1}$, $s_i \subepostar t_i$, $s \subepostar t_{i+1}, \dots, s \subepostar t_{k}$, and
      \item $s \epostar t_{k+1}, \dots, s \epostar t_{k+n}$.
      \end{enumerate}
  \end{enumerate}
  Here we set ${\epostareq} \defi {\epostar} \cup {\eqis}$.
\end{definition}

\begin{theorem}\label{th:epostar:sound}
  Suppose $\RS$ is a constructor TRS compatible with $\epostar$,i.e., $\RS \subseteq {\epostar}$. 
  Then the innermost runtime complexity $\rcRi(n)$ is bounded by an 
  exponential $2^{\bigO(n^k)}$ for some fixed $k \in \NAT$.
\end{theorem}
We prove Theorem \ref{th:epostar:sound} in Section~\ref{s:soundness}.

\begin{example}\label{ex:2}[Example \ref{ex:1} continued].
  Let $\safe$ be the safe mapping such that 
  $\safe(\ffib) = \varnothing$ and $\safe(\dfib) = \{2\}$.
  Further, let $\qp$ be the admissible precedence 
  with $\ffib \sp \dfib \sp \ms \ep 0$. Then one verifies that 
  $\RSfib \subseteq {\epostar}$ for the induced order $\epostar$. 
  By Theorem~\ref{th:epostar:sound} we conclude that the 
  innermost runtime complexity of $\RSfib$ is exponentially bounded.
\end{example}

Define the \emph{derivational complexity} of a rewrite system $\RS$ 
as $\dcR(n) \defi \max \{\dl(t,\to) \mid t \in \TERMS \text{ and } \size{t} \leqslant n\}$.
The following example demonstrates that Theorem~\ref{th:epostar:sound} 
\emph{does neither} hold for full rewriting \emph{nor} derivational complexity.
\begin{example}
  Consider the TRS $\RSd$ consisting of the rules
  $$\m{d}(;x) \to \m{c}(;x,x) \qquad \m{f}(0;y) \to y \qquad \m{f}(\m{s}(;x);y) \to \m{f}(x;\m{d}(;\m{f}(x;y)))\tpkt$$
  Then $\RSd \subseteq {\epostar}$ for the precedence $\m{f} \sp \m{d} \sp \m{c}$
  and safe mapping as indicated in the definition of $\RSd$.  
  Theorem~\ref{th:epostar:sound} proves that the 
  innermost runtime complexity of $\RSd$ is exponentially bounded.
\end{example}
On the other hand, the runtime complexity of $\RSd$ 
(with respect to full rewriting) grows strictly faster 
than any exponential:
Consider for arbitrary $t \in \TERMS$ the term $f(s^n(0),t)$. 
We verify, for $n > 0$, $\dl(f(s^n(0),t),\rew) \geqslant 2^{2^{n-1}} \cdot (1 + \dl(t,\rew))$
by induction on $n$.
For $m \in \NAT$, set $\numeral{m} \defi \m{s}^m(0)$.
Consider the base case $n=1$.
Then any maximal derivation 
$$
\m{f}(\numeral{1},t) 
\rew \m{f}(0,\m{d}(\m{f}(0,t))) 
\rew \m{f}(0,\m{c}(\m{f}(0,t),\m{f}(0,t)))
\rsl{3} \m{c}(t,t) \rew \cdots
$$
proves this case. For this observe that $\dl(\m{c}(t,t), \rew) = 2 \cdot \dl(t,\rew)$, 
and hence $\dl(\m{f}(\numeral{1},t),t) \geqslant 5 + 2 \cdot \dl(t,\rew) > 2^{2^{0}} \cdot (1 + \dl(t,\rew))$.
Notice that we employ lazy reduction of $\m{d}$ in an essential way.
For the inductive step, consider a maximal derivation
$\m{f}(\numeral{n+1},t) \rew \m{f}(\numeral{n},\m{d}(\m{f}(\numeral{n},t))) \rew \cdots$.
Applying induction hypothesis twice we obtain
\begin{align*}
  \dl(\m{f}(\numeral{n+1},t), \rew) > \dl(\m{f}(\numeral{n},\m{d}(\m{f}(\numeral{n},t))), \rew) 
  & > \dl(\m{f}(\numeral{n},\m{f}(\numeral{n},t)), \rew) \\
  & > 2^{2^{n-1}} \cdot (2^{2^{n-1}} \cdot (1 + \dl(t,\rew))) \\
  & = 2^{2^{n}} + 2^{2^{n}} \cdot \dl(t,\rew) \tpkt
\end{align*}
\qed


We now present the application of Theorem~\ref{th:epostar:sound}
in the context of \emph{implicit computational complexity (ICC)}.
Following \cite{BMM09}, and extended to nondeterministic computation 
in \cite{AM10b,BM10}, we give semantics to TRS $\RS$ as follows:
\begin{definition}
  Let $\Val \defi \TA(\CS,\VS)$ denote the set of \emph{values}.
  Further, let $\NA \subseteq \Val$ be a finite set of
  \emph{non-accepting patterns}.  We call a term $t$ \emph{accepting}
  (with respect to $\NA$) if there exists no $p \in \NA$ such that
  $p\sigma = t$ for some substitution $\sigma$.  We say that $\RS$
  \emph{computes the relation $R \subseteq {\Val \times \Val}$} with
  respect to $\NA$ if there exists $\m{f} \in \DS$ such that for all
  $s,t \in \Val$,
  \begin{equation*}
    {s \RR t} \text{\qquad iff \qquad} {\m{f}(s) \irsn[\RS] t} \text{ and $t$ is accepting}\tpkt
  \end{equation*}
  On the other hand, we say that a relation $R$ is computed by $\RS$
  if $R$ is defined by the above equations with respect to \emph{some}
  set $\NA$ of non-accepting patterns.
\end{definition}
For the case that $\RS$ is \emph{confluent}
we also say that $\RS$ computes the (partial) \emph{function} 
induced by the relation $R$. Note that the
restriction to binary relations is a non-essential simplification.
The assertion that for normal forms $t$, $t$ is accepting
aims to eliminate by-products of the computation that should not be
considered as part of the computed relation $R$.

As a consequence of Theorem \ref{th:epostar:sound} we derive our main result.
Following \cite{LM09,AM10b} we employ \emph{graph rewriting} \cite{P:01} 
to efficiently compute normal forms.

\begin{theorem}[Soundness]\label{t:soundness}
  Suppose $\RS$ is a constructor TRS compatible with $\epostar$. 
  The relations computed by $\RS$ are computable in nondeterministic time $2^{\bigO(n^k)}$
  for some $k \in \NAT$.
  In particular, if $\RS$ is confluent then
  $f \in \FEXP$ for each function $f$ computed by $\RS$.
\end{theorem}
\begin{proof}
  We sketch the implementation of the relation $R_f$ (function $f$) on a 
  Turing machine $\M_f$.
  Single out the corresponding defined function symbol $\m{f}$, and
  consider some arbitrary input $v \in \Val$. 
  First writing $\m{f}(v)$ on a dedicated working tape, 
  the machine $\M_f$ iteratively rewrites $\m{f}(v)$ to normal 
  form in an innermost fashion. For non-confluent TRSs $\RS$, 
  the choice of the redex is performed nondeterministically, 
  otherwise some innermost redex is computed deterministically.
  By the assumption $\RS \subseteq {\epostar}$, Theorem~\ref{th:epostar:sound}
  provides an upper bound $2^{\size{\m{f}(v)}^{c_1}}$
  on the number of iterations for some $c_1 \in \NAT$, i.e., 
  the machine performs at most exponentially many iterations in the size
  of the input $v$. 
%
  To investigate into the complexity of a single iteration, 
  consider the $i$-th iteration
  with $t_i$ written on the working tape (where $\m{f}(v) \rsl{i} t_i$). 
  We want to compute some $t_{i+1}$ with $t_i \irew t_{i+1}$. 
  Observe that in the presence of duplicating rules, $\size{t_i}$ might 
  be exponential in $i$ (and $\size{v}$).
  As we can only assume $i \leqslant 2^{\size{\m{f}(v)}^{c_1}}$, 
  we cannot hope to construct $t_{i+1}$ from $t_i$ in time exponential in $\size{v}$
  if we use a representation of terms that is linear in size in the number of symbols.
  Instead, we employ the machinery of \cite{AM10b}. 
  By taking sharing into account, 
  \cite{AM10b} achieves an encoding of $t_i$ that is bounded in size polynomially in $\size{v}$ and $i$.
  Hence in particular $t_i$ is encoded in size $2^{\size{s}^{c_2}}$ for 
  some $c_2 \in \NAT$ depending only on $\RS$.
  In the setting of \cite{AM10b} a single step is computable in polynomial time
  (in the encoding size). And so $t_{i+1}$ is
  computable from $t_i$ in time $2^{\size{s}^{c_3}}$ for some $c_3 \in \NAT$
  depending only on $\RS$.
  Overall, we conclude that normal forms are computable in time 
  $2^{\size{s}^{c_1}} \cdot 2^{\size{s}^{c_3}} = 2^{\bigO(\size{s}^k)}$ for some $k \in \NAT$ worst case.
  After the final iteration, the machine $\M_f$ checks whether the computed normal form $t_l$ is accepting
  and either accepts or rejects the computation.
  Using the machinery of \cite{AM10b} pattern matching is polynomial the encoding size of
  $t_l$, by the above bound on encoding sizes the operation is exponential in $\size{v}$.
  As $v$ was chosen arbitrary and $k$ depends only on $\RS$, we conclude the theorem.
\end{proof}

In correspondence to Theorem~\ref{t:soundness}, 
\EPOSTAR~is complete in the following sense.
Again this is proved in a separate section below (c.f. Section~\ref{s:completeness}).
\begin{theorem}[Completeness]\label{th:epostar:complete}
  Suppose $f \in \FEXP$. Then there exists a confluent, constructor TRS $\RS_f$ computing $f$ 
  that is compatible with some exponential path order $\epostar$.
\end{theorem}


\subsection{Soundness}\label{s:soundness}

We now prove Theorem~\ref{th:epostar:sound}, frequently employing the following:
\begin{lemma}\ \label{l:epostar:help}
  The inclusions ${\subepostar} \subseteq {\esupertermstrict} \subseteq {\epostar}$ hold and further, 
  if $s \in \TA(\CS,\VS)$ and $s \epostar t$ then $t \in \TA(\CS,\VS)$. 
\end{lemma}
\begin{proof}
  Both properties are straight forward consequences of Definition~\ref{d:subepo} and Definition~\ref{d:epostar}.
  For the second property we require that the precedence $\qp$ is admissible. 
  One easily verifies
  that if $t \not \in \TA(\CS,\VS)$, then $\qp$ is \emph{not} admissible.
\end{proof}

Let $\RS$ be a TRS compatible with some instance $\epostar$.
The idea behind the proof of Theorem \ref{th:epostar:sound} is 
to translate $\irew$-derivations into $\epo[\ell]$-descents 
for some fixed $\ell \in \NAT$ depending only $\RS$. Once this translation is established, 
we can use Theorem \ref{t:epo} to bind the runtime-complexity of $\RS$ appropriately.
For the moment, suppose $\RS$ is completely defined. 
We replace this restriction by constructor TRS later on.
Since $\RS$ is completely defined, normal forms and 
constructor terms coincide, and thus $s \irew t$ if $s = C[l\sigma], t = C[r\sigma]$ for some 
rule ${l \to r} \in \RS$ where additionally $l\sigma \in \Tb$.
Let $t$ be obtained by rewriting a basic term $s$. By the use of $\subepostar$ in Definition 
\ref{d:epostar} every normal argument $t_i$ of $t$ is irreducible, i.e., $t_i \in \TA(\CS,\VS)$. 
We capture this observation in the definition of $\Tn$:
\begin{definition}
  The set $\Tn$ 
  is the least set of terms such that 
  (i) $\TA(\CS,\VS) \subseteq \Tn$, and 
  (ii) if $f \in \FS$, $\vec{s} \subseteq \TA(\CS,\VS)$ and $\vec{t} \subseteq \Tn$
    then $f(\sn{\vec{s}}{\vec{t}}) \in \Tn$.
\end{definition}
Note that $\Tb \subseteq \Tn$. 
\begin{lemma}\label{l:epostar:tnderiv}
  Let $\RS$ be a completely defined TRS compatible with $\epostar$, 
  and let $s \in \Tn$.
  If $s \irew t$ then $t \in \Tn$.
\end{lemma}
\begin{proof}
  Suppose $s \irew t$ where $s \in \Tn$, i.e., there exists a rule 
  ${l \to r} \in \RS$ such that
  $s = C[l\sigma]$, $t = C[r\sigma]$, and for all direct subterms $l_i$ of $l$, 
  $l_i\sigma \in \NF(\RS)$.
  As $\RS$ is completely defined $\NF(\RS) = \TA(\CS,\VS)$. We conclude
  $l \in \Tb$ and $\ofdom{\sigma}{\VS \to \TA(\CS,\VS)}$.
  Since $s \in \Tn$, it follows that $t \in \Tn$ if $r\sigma \in \Tn$ by definition of $\Tn$
  and the fact $l\sigma \not\in \TA(\CS,\VS)$.
  Note that $\Tn$ is closed under substitutions with image in $\TA(\CS,\VS)$, in 
  particular $r\sigma \in \Tn$ follows if $r \in \Tn$.
  We prove the latter by side induction on $l \epostar r$. 

  If $l_i \epostareq r$ for some direct subterm $l_i$ of $l$ then $r \in \TA(\CS,\VS)$ 
  by Lemma \ref{l:epostar:help} as $l \in \Tb$.
  Next, suppose either Definition \ref{d:epostar}.\ref{d:epostar:2} or Definition \ref{d:epostar}.\ref{d:epostar:3}
  applies.
  Then, by definition, $r=g(\sn{r_1, \dots, r_k}{r_{k+1}, \dots r_{k+n}})$ for some $g \in \FS$.
  For $i \in \{1,\dots,k \}$, 
  $l \subepostar r_i$ follows from 
  Definition \ref{d:epostar}.\ref{d:epostar:2} \emph{and} Definition \ref{d:epostar}.\ref{d:epostar:3}.
  Consequently $r_i \in \TA(\CS,\VS)$ employing $l \in \Tb$ and ${\subepostar} \subseteq {\esupertermstrict}$ (c.f. Lemma~\ref{l:epostar:help}).
  For $i \in \{k+1,\dots,k+n \}$ we observe $l \epostar r_i$. 
  Induction hypothesis yields  $r_i \in \Tn$. 
  We conclude $r \in \Tn$ by definition of $\Tn$.
\end{proof}

We embed $\irew$-steps in $\epo[\ell]$ using \emph{predicative interpretations} $\I$.
Lemma \ref{l:epostar:tnderiv} justifies that we only consider terms from $\Tn$.
For each defined symbol $f$, let $\fn$ be a fresh function symbol, and 
let $\FSn = \{ \fn \mid f \in \DS \} \cup \CS$. Here the arity of $\fn$
is $k$ where $\normal(f) = \{\seq[k]{i}\}$, moreover
$\fn$ is still considered a defined function symbol when applying Definition \ref{d:epo}.
We further extend the (admissible) precedence $\qp$ to $\FSn$ 
in the most obvious way: $\fn \ep \gn$ if $f \ep g$ and $\fn \sp \gn$ if $f \sp g$.
\begin{definition}\label{d:epostar:N}
  A \emph{predicative interpretation} $\I$ is a mapping $\ofdom{\I}{\Tn \to \TAL}$ 
  defined as follows:
  \begin{enumerate}
  \item $\I(t) = \nil$ if $t \in \TA(\CS,\VS)$, and otherwise
  \item $\I(t) = \lst{\fn(\seq[k]{t})} \append \I(t_{k+1}) \append \cdots \append \I(t_{k+n})$ 
  for $t = f(\pseq[k][n]{t})$. 
  \end{enumerate}
\end{definition}

The next lemma provides the embedding of root steps for completely defined, compatible, TRSs $\RS$.
Here we could simply define $\I(t) = \fn(\seq[k]{t})$ in case (ii).
The complete definition becomes only essential when we look at closure under context
in Lemma~\ref{l:epostar:embedctx} below.
\begin{lemma}\label{l:epostar:embedroot}
  Let $s \in \Tb$ and let $\ofdom{\sigma}{\VS \to \TA(\CS,\VS)}$ be a substitution. 
  If $s \epostar t$ then $\I(s\sigma) \epo[\size{t}] \I(t\sigma)$.
\end{lemma}
\begin{proof}
  By the assumptions, 
  $\I(s\sigma) = \lst{\fn(s_1\sigma, \dots, s_l\sigma)} = \fn(s_1\sigma, \dots, s_l\sigma)$ for $f$ the (defined) root symbol of $s$ and 
  normal arguments $s_i$ of $s$.
  If $t \in \TA(\CS,\VS)$ then the lemma trivially follows
  as $\I(t\sigma) = \nil$.
  We prove the remaining cases by induction on the definition of $\epostar$, 
  thus we have $s \epostar t$ either by Definition~\ref{d:epostar}.\ref{d:epostar:2} or 
  Definition~\ref{d:epostar}.\ref{d:epostar:3}.
  Let $t = g(\pseq[k][n]{t})$ and so
  $$\I(t\sigma) = \lst{\gn(t_1\sigma, \dots, t_k\sigma)} 
  \append \I(t_{k+1}\sigma) \append \cdots \append \I(t_{k+n}\sigma) \tpkt$$ 
  Observe that $\I(x\sigma) = \nil$ for all variables $x$ in $t$.
  Using this we see that the length of the list
  $\I(t\sigma)$ is bound by $\size{t}$. 
  Hence by Definition~\ref{d:epo}.\ref{d:epo:2}, 
  it suffices to verify $\I(s\sigma) \epo[\size{t}] \I(t_i\sigma)$ for all safe arguments 
  $t_i$ ($i\in\{k+1,\dots,m\}$), and further
  \begin{equation}
    \label{l:epostar:embedroot:eq1}
    \fn(s_1\sigma, \dots, s_l\sigma) \epo[\size{t}] \gn(t_1\sigma, \dots, t_k\sigma) \tpkt
  \end{equation}
  As we have $s \epostar t_i$ on safe argument $t_i$, 
  the former follow by induction hypothesis on the terms $t_i$.
  It remains to verify \eqref{l:epostar:embedroot:eq1}. 
  We continue by case analysis. 
  \begin{enumerate}
  \item Suppose $f \sp g$, i.e., Definition~\ref{d:epostar}.\ref{d:epostar:2} applies. 
  Then $\fn \sp \gn$ by definition. By Definition~\ref{d:epo}.\ref{d:epo:3} it suffices to prove 
  $\fn(s_1\sigma, \dots, s_l\sigma) \esupertermstrict t_i\sigma$ for all $i \in \{1,\dots,k\}$.
  Fix $i \in \{1,\dots,k\}$.
  According to Definition~\ref{d:epostar}.\ref{d:epostar:2} $s \subepo t_i$ holds, 
  and thus there exists $j \in \{1,\dots,l\}$ such that $s_j \subepostareq t_i$. 
  Hence $s_j \esuperterm t_i$ by Lemma~\ref{l:epostar:help}, 
  from which we conclude $\fn(s_1\sigma, \dots, s_l\sigma) \esupertermstrict t_i\sigma$ 
  since we suppose $\ofdom{\sigma}{\VS \to \TA(\CS,\VS)}$.

  \item   Suppose $\fn \ep \gn$, i.e., Definition~\ref{d:epostar}.\ref{d:epostar:3} applies. 
  By Definition~\ref{d:epo}.\ref{d:epo:5} it suffices to prove
  (i) $s_1\sigma \eqi t_1\sigma$, \dots, $s_{\ell - 1}\sigma \eqi t_{\ell - 1}\sigma$, 
  (ii) $s_\ell\sigma \esupertermstrict t_\ell\sigma$, and further
  (iii) $\fn(s_1\sigma, \dots, s_l\sigma) \esupertermstrict t_{\ell+1}\sigma$, \dots, $\fn(s_1\sigma, \dots, s_l\sigma) \esupertermstrict t_{k}\sigma$
  for some $\ell \in \{1,\dots,k\}$.
  The assumptions in Definition~\ref{d:epostar}.\ref{d:epostar:3}
  yield $s_1 \eqis t_1$, \dots, $s_{\ell-1} \eqis t_{\ell-1}$ from which we conclude (i), 
  further $s_\ell \subepostar t_\ell$ from which we conclude (ii) with the help of Lemma~\ref{l:epostar:help} 
  (using $s_\ell \in \TA(\CS,\VS)$),
  and finally $s \subepostar t_{\ell+1}$, \dots, $s \subepostar t_{k}$ 
  from which we obtain (iii) as in the case above. 
  
  \end{enumerate}
\end{proof}

\begin{lemma}\label{l:epostar:embedctx}
  Let $s,t \in \Tn$ and let $C$ be a context such that $C[s] \in \Tn$.
  If $\I(s) \epo[\ell] \I(t)$ then $\I(C[s]) \epo[\ell] \I(C[t])$.
\end{lemma}
\begin{proof}
We show the lemma by induction on $C$.
It suffices to consider the step case.
Observe that by the assumption $\I(s) \epo[\ell] \I(t)$, $s \not \in \TA(\CS,\VS)$
since otherwise $\I(s) = \nil$ is $\epo[\ell]$-minimal.
We can thus assume
$
C = f(\sn{\seq[k]{s}}{s_{k+1}, \dots, C'[\hole], \dots s_{k+l}})
$
for some context $C'$ by definition of $\Tn$.
Thus, for each $u \in \{ s, t \}$,
$$
\I(C[u]) = \lst{\fn(\seq[k]{s})} 
           \append \I(s_{k+1}) \append \cdots \append \I(C'[u]) \append \cdots \append \I(s_{k+l}) \tpkt
$$
By induction hypothesis $\I(C'[s]) \epo[\ell] \I(C'[t])$. We conclude
using Lemma \ref{permlem}.
\end{proof}

\begin{lemma}\label{l:epostar:embed}
  Let $\RS$ be a completely defined TRS compatible with $\epostar$. Let $s \in \Tn$.
  If $s \irew t$ then $\I(s) \epo[\ell] N(t)$ where $\ell \defi \max\{ \size{r} \mid {l \to r} \in \RS\}$.
\end{lemma}
\begin{proof}
  Suppose $s \irew t$. Hence there exists a context $C$, substitution 
  $\sigma$ and rule ${l \to r} \in \RS$ such that 
  $s = C[l\sigma]$ and $t = C[r\sigma]$. By the assumption that $\RS$ is completely defined, 
  $l \in \Tb$ and $\ofdom{\sigma}{\VS \to \TA(\CS,\VS)}$. 
  Since $\RS \subseteq {\epostar}$, we obtain $\I(l\sigma) \epo[\ell] \I(r\sigma)$ by Lemma~\ref{l:epostar:embedroot}
  (additionally employing ${\epo[\size{r}]} \subseteq {\epo[\ell]}$).
  Lemma~\ref{l:epostar:embedctx} then establishes $\I(s) \epo[\ell] \I(t)$.
\end{proof}

\begin{theorem}\label{t:epostar:main:cd}
  Let $\RS$ be a completely defined, possibly infinite, TRS compatible with $\epostar$. 
  Suppose $\ell \defi \max\{ \size{r} \mid {l \to r} \in \RS\}$ is well-defined.
  There exists $k \in \NAT$ such that $\rcRi(n) \leqslant 2^{\bigO(n^k)}$.
\end{theorem}
\begin{proof}
  We prove the existence of $c_1,c_2 \in \NAT$ so that for any 
  $s \in \Tb$, $\dl(s, \irew) \leqslant 2^{c_1 \cdot \size{s}^{c_2}}$.
  Consider some maximal derivation
  $s = t_0 \irew t_1 \irew \cdots \irew t_{n}$.
  Let $i \in \{0, \dots, n-1\}$.
  We observed $t_i \in \Tn$ in Lemma~\ref{l:epostar:tnderiv},
  and thus $\I(t_i) \epo[\ell] \I(t_{i+i})$ due to Lemma~\ref{l:epostar:embed}.
  So in particular $\dl(s,\irew) \leqslant \Slow[\ell](\I(s))$.
  We estimate $\Slow[\ell](\I(s))$ in terms of $\size{s}$:
  for this, suppose $s = f(\pseq{s})$ for some $f \in \DS$ and $s_i \in \TA(\CS,\VS)$ ($i \in \{1, \dots, k + l\}$).
  By definition $\I(s) = \fn(\seq[k]{s})$.
  Set $N \defi \max \{ \Slow[\ell] ( s_i ) \mid 1 \leqslant i \leqslant k \} +1$, and
  verify

  \begin{equation}
    \label{eq:epostar:main:cd:1}
    N 
    \leqslant 1 + \sum_{ i=1 }^k \Slow[\ell](s_i)
    \leqslant 1 + \sum_{ i=1 }^k \depth(s_i) 
    \leqslant \size{s}\tpkt
  \end{equation}

  For the second inequality we employ Lemma~\ref{lemG_k}, which gives
  $\Slow[\ell](s_i) = \depth(s_i)$ as $s_i \in \TA(\CS,\VS)$ for all
  $i \in \{1, \dots, k\}$.
  Applying Theorem~\ref{t:epo} we see

  \begin{align*}
    \Slow[\ell](\I(s)) & = \Slow[\ell](\fn(\seq[k]{s})) \\
    & \leqslant (\ell+1)^{N^\ell \cdot \rk(\fn) + \sum_{ i=1 }^k N^{ \ell-i } \cdot \Slow[\ell] ( s_i ) } & (\text{by Theorem~\ref{t:epo}}) \\
    & \leqslant (\ell+1)^{\size{s}^\ell \cdot \rk(\fn) + \size{s}^\ell \cdot \sum_{ i=1 }^k \Slow[\ell](s_i) } & (\text{by Equation~\ref{eq:epostar:main:cd:1}})\\
    & \leqslant (\ell+1)^{\size{s}^\ell \cdot \rk(\fn) + \size{s}^{\ell} \cdot \size{s}}  & (\text{by Equation~\ref{eq:epostar:main:cd:1}})\\
    & \leqslant (\ell+1)^{(\rk(\fn) + 1) \cdot \size{s}^{\ell+1}} \tpkt
  \end{align*}

  Since $\ell$ depends only on $\RS$, and $\rk(\fn)$ is bounded by some constant depending only on $\FS$,
  simple arithmetical reasoning gives the
  constants $c_1,c_2$ such that 
  $\dl(s,\irew) \leqslant \Slow[\ell](\I(s)) \leqslant 2^{c_1 \cdot \size{s}^{c_2}}$. This concludes the Theorem.
\end{proof}

We now lift the restriction that $\RS$ is completely defined for constructor TRSs $\RS$.
The idea is to extend $\RS$ with sufficiently many rules 
so that the resulting system is completely defined and Theorem~\ref{t:epostar:main:cd} applicable.
\begin{definition}\label{d:epostar:rss}
  Let $\bot$ be a fresh constructor symbol and $\RS$ a TRS.
  We define $\RSS \defi \{t \to \bot \mid t \in \TA(\FS \cup \{\bot\},\VS) \cap \NF(\RS) \text{ and the root symbol of $t$ is defined} \}$.
\end{definition}
We extend the precedence $\qp$ to $\FS \cup \{\bot\}$ so that $\bot$ is minimal.
Thus $\RSS \subseteq {\epostar}$ follows by one application of 
Definition \ref{d:epostar}.\ref{d:epostar:2}.
Further, the \emph{completely defined} TRS $\RS \cup \RSS$ is able to 
simulate $\irew$ derivations for constructor TRS $\RS$:

\begin{lemma}\label{l:epostar:simul}
  Suppose $\RS$ is a constructor TRS. Then $\RS \cup \RSS$ is completely defined.
  Further, if $s \irsl[\RS]{\ell} t$ then $s \irsl[\RS \cup \RSS]{\ell'} t'$ for some $t'$ and $\ell' \geqslant \ell$.
\end{lemma}
\begin{proof}
  That $\RS \cup \RSS$ is completely defined follows by definition. 
  We outline the proof of the second statement.
  For a complete proof we kindly refer the reader to \cite[Section 5.1]{MA09}.
  Let $\normalise{t}$ denote the unique normal form of $t \in \TA(\FS \cup \{\bot\},\VS)$
  with respect to $\RSS$ (observe that $\RSS$ is confluent and terminating by definition).
  One verifies that for ${l \to r} \in \RS$, $\ofdom{\sigma}{\VS \to \NF(\RS)}$
  and $\sigma_{\normalise{}} \defi \{x \mapsto \normalise{u} \mid \sigma(x) = u \}$, 
  \begin{align}
    \label{l:epostar:simul:eq1}
    \normalise{(l\sigma)} = l{\sigma_{\normalise{}}} \irew[\RS \cup \RSS] r{\sigma_{\normalise{}}} \irss[\RS \cup \RSS] \normalise{(r\sigma)} \tpkt  
  \end{align}
  Using equation \eqref{l:epostar:simul:eq1}, we obtain 
  $\normalise{s} \irst[\RS \cup \RSS] \normalise{t}$ from $s \irew[\RS] t$
  by a straight forward inductive argument.
  It is not difficult to see that from this we can conclude the lemma.
\end{proof}

An immediate consequence of Lemma~\ref{l:epostar:simul}
is $\rcRi[\RS](n) \leqslant \rcRi[\RS \cup \RSS](n)$, i.e., 
the innermost runtime-complexity of $\RS$ can be analysed 
through $\RS \cup \RSS$.
We arrive at the proof of our main theorem:
\begin{proof}[Proof of Theorem~\ref{th:epostar:sound}]
  Suppose $\RS$ is a constructor TRS compatible with $\epostar$. We verify that
  $\rcRi(n)$ is bounded by an exponential $2^{\bigO(n^k)}$ for some fixed $k \in \NAT$:
  let $\RSS$ be defined according to Definition~\ref{d:epostar:rss}.
  By Lemma~\ref{l:epostar:simul}, ${\RS \cup \RSS}$ is completely defined, 
  and moreover, $\rcRi[\RS](n) \leqslant \rcRi[\RS \cup \RSS](n)$.
  Clearly $\max\{\size{r} \mid {l \to r} \in \RSS \} = 1$, since 
  $\RS$ is finite we have that $\max\{\size{r} \mid {l \to r} \in {\RS \cup \RSS} \}$ is well-defined.
  Further ${(\RS \cup \RSS)} \subseteq {\epostar}$ follows by the assumption on $\RS$
  and definition of $\RSS$.
  Hence all assumptions of Theorem~\ref{t:epostar:main:cd} are fulfilled, 
  and we conclude $\rcRi[\RS](n) \leqslant \rcRi[\RS \cup \RSS](n) \leqslant 2^{\bigO(n^k)}$ 
  for some $k \in \NAT$.
\end{proof}


\subsection{Completeness}\label{s:completeness}

To prove Theorem~\ref{th:epostar:complete}, 
we use the characterisation of the exponential time 
computable functions given in \cite{AraiE09} by Arai and the second author, 
and, the resulting \emph{term rewriting characterisation} 
given in \cite{Eguchi09}.
We closely follow the presentation of \cite{AraiE09,Eguchi09}, 
for further motivation of the presented notions we kindly 
refer the reader to \cite{AraiE09,Eguchi09}.

In the spirit of \cite{BC92}, the class $\NC$ (of functions over binary words) from \cite{AraiE09}
relies on a syntactic separation of argument positions into \emph{normal} and \emph{safe} ones.
To highlight this separation, we write $f(\vec{x};\vec{y})$ instead of 
$f(\vec{x},\vec{y})$ for normal arguments $\vec{x}$ and safe arguments $\vec{y}$. 
The class $\NC$ is defined as the least class
containing certain initial functions and that is closed under 
the scheme of
\emph{(weak) safe composition}
\begin{equation}\label{eq:pcompm} \tag{$\m{WSC}$}
  f(\vec{x};\vec{y}) = h(x_{i_1},\dots,x_{i_k}; \vec{s}(\vec{x};\vec{y})) \tkom
\end{equation}
and \emph{safe nested recursion on notation}
\begin{equation} \label{eq:snrs} \tag{$\m{SNRN}$}
  \begin{array}{r@{~}l}
      f(\vec{\varepsilon},\vec{x};\vec{y}) & =  g(\vec{x};\vec{y}) \\
      f(\vec{z},\vec{x};\vec{y}) & = 
      h_{\tau(\vec{z})}(\vec{v_1},\vec{x};\vec{y},f(\vec{v_1},\vec{x};\vec{t}_{\tau(\vec{z})}(\vec{v_2}, \vec{x}; \vec{y},f(\vec{v_2},\vec{x};\vec{y}))))
  \end{array}
\end{equation}
where $\vec{z} \not=\vec{\varepsilon}$.
The Scheme~\eqref{eq:pcompm} reflects that the exponential 
time functions are \emph{not} closed under composition.
We have presented the Scheme~\eqref{eq:snrs} with two 
nested recursive calls for brevity, however \cite{AraiE09}
allows an arbitrary (but fixed) number of nestings.
Note that here recursion
is performed simultaneously on multiple arguments $\vec{z}$. 
The functions $h_{\tau(\vec{z})}$ and $\vec{t}_{\tau(\vec{z})}$ are 
previously defined functions, chosen in terms of $\tau(\vec{z}) \in \Sigma^k_0$.
Here $k$ equals the length 
of $\vec{z}$, and $\Sigma^k_0 \defi \{0,1,\varepsilon \}^k \setminus \{ \varepsilon \}^k$.
Further, $\vec{v_1}$ and $\vec{v_2}$ are unique predecessors
of $\vec{z}$ defined in terms of $\tau(\vec{z})$.
In \cite{AraiE09} it is proved that $\NC$ coincides with $\FEXP$.

The term rewriting characterisation from \cite{Eguchi09} expresses
the definition of $\NC$ as an \emph{infinite} rewrite system $\RN$,
depicted below.
Here binary words are formed from the constructor symbols $\varepsilon$, $\mS_0$
and $\mS_1$. For notational reasons we use $\mS_\varepsilon(;z)$ to 
denote $\varepsilon$.
The function symbols $\ZEROP{k}{l}, \PROJP{k}{l}{r}, \PREDP, \CONDP$
correspond to the initial functions of $\NC$. 
The symbol $\SUBP[g,\seq[k]{i},\vec{h}]$ is used to denote the
function obtained by composing functions $g$ and $\vec{h}$ according to
the Scheme~\eqref{eq:pcompm}. 
Finally, the function symbol
$\SNRN[g,h_{w},\vec{s_{w}},\vec{t_{w}}\,(w \in \Sigma^k_0)]$ corresponds
to the function defined by safe nested recursion on notation from 
$g$, $h_w$, $\vec{s_w}$, $\vec{t_w}$ $(w \in \Sigma^k_0)$ in accordance
to Scheme~\eqref{eq:snrs}.  
We highlight the separation of safe and normal argument positions
directly in the rules. The TRS $\RN$ consists of the rules
$$
\begin{array}{l@{\qquad}l}
  \ZEROP{k}{l}(\vec{x};\vec{y}) \to \varepsilon 
  & \PREDP(;\varepsilon) \to \varepsilon \\
  \TOP \PROJP{k}{l}{r}(\vec{x};\vec{y}) \to x_r \text{ for $r \in \{1,\dots,k \}$} 
  & \PREDP(;\mS_i(;x)) \to x \\
  \TOP \PROJP{k}{l}{r}(\vec{x};\vec{y}) \to y_{r-k} \text{ for $r \in \{k+1, \dots, l+k\}$} 
  & \CONDP(;\varepsilon,y_0,y_1) \to y_0 \\
  \TOP \SUBP[g,\seq[k]{i},\vec{h}](\vec{x};\vec{y}) \to g(x_{i_1},\dots,x_{i_k};\vec{h}(\vec{x};\vec{y})) 
  & \CONDP(;\mS_i(;x),y_0,y_1) \to y_i\\
  \multicolumn{2}{l}{
    \TOP \SNRN[g,h_{w},\vec{s_{w}},\vec{t_{w}}\,(w \in \Sigma^k_0)](\vec{\varepsilon}, \vec{x};\vec{y}) \to g(\vec{x};\vec{y})
  } \\
  \multicolumn{2}{l}{  
    \TOP \SNRN[g,h_{w},\vec{s_{w}},\vec{t_{w}}\,(w \in \Sigma^k_0)](\mS_{i_1}(;z_1), \dots, \mS_{i_k}(;z_k), \vec{x};\vec{y}) \to {}
  }\\
  \multicolumn{2}{l}{  
    \TOP \qquad h_{i_1 \cdots i_k}(\vec{v_1},\vec{x};\vec{y}, \SNRN[g,h_{w},\vec{s_{w}},\vec{t_{w}}\,(w \in \Sigma^k_0)](\vec{v_1},\vec{x};\vec{a}))
  }\\
  \multicolumn{2}{l}{  
    \TOP \qquad\quad [\vec{s}_{i_1 \cdots i_k}(\vec{v_2},\vec{x};\vec{y}, \SNRN[g,h_{w},\vec{s_{w}},\vec{t_{w}}\,(w \in \Sigma^k_0)](\vec{v_2},\vec{x};\vec{b})) \slash \vec{a}]
  } \\
  \multicolumn{2}{l}{  
    \TOP \qquad\quad [\vec{t}_{i_1 \cdots i_k}(\vec{v_3},\vec{x};\vec{y}, \SNRN[g,h_{w},\vec{s_{w}},\vec{t_{w}}\,(w \in \Sigma^k_0)](\vec{v_3},\vec{x};\vec{y})) \slash \vec{b}]
  } \\
  \multicolumn{2}{l}{  
    \TOP \hspace{5.7cm} \text{where $\{ \varepsilon \}^k \neq \{ i_j \mid 1 \leqslant j \leqslant k\} \subseteq \{\varepsilon, 0,1\}^k$.}
  }
\end{array}
$$
Abbreviate $\vec{u} = u_1,\dots,u_k = \mS_{i_1}(;z_1), \dots, \mS_{i_k}(;z_k)$, 
and consider for some $j \in \{1,2,3\}$ arguments  $\vec{v_j} = {v_1, \dots, v_k}$.
The arguments $\vec{v_j}$ are \emph{$\succ$-predecessors} \cite{AraiE09} 
of $\vec{u}$.
This gives some $i \in \{1, \dots,  k\}$ such that
(i) $u_1 = v_1, \dots, u_{i-1} = v_{i-1}$, (ii) $u_i \subepostar v_{i}$
and (iii) $u_{l_{i+1}} \subepostareq v_{i+1}, \dots, u_{l_{k}} \subepostareq v_{k}$
for some $l_{i+1}, \dots, l_k \in \{1,\dots,k\}$.

By the results from \cite{Eguchi09}, it follows
that for each function $f$ from $\FEXP$ there exists a \emph{finite restriction} $\RS_f$ of 
$\RN$ which computes the function $f$.
Hence to prove Theorem~\ref{th:epostar:complete}, 
it suffices to orient each finite restriction of $\RN$ by
an instance of $\EPOSTAR$.

\begin{proof}[Proof of Theorem~\ref{th:epostar:complete}]
  Consider some arbitrary function $f \in \FEXP$  
  and the corresponding TRS $\RS_f \subseteq \RN$ computing $f$.
  Let $\FS$ be the signature consisting of function symbols appearing in
  $\RS_f$.
  For function symbols $g,h \in \FS$, we define 
  $g \sp h$ in the precedence iff $\lh(g) > \lh(h)$, where
  \begin{enumerate}
  \item $\lh(g) := 1$ for $g \in \{\ZEROP{k}{l}, \PROJP{k}{l}{r}, \PREDP, \CONDP \}$, 
  \item $\lh(\SUBP[g,\seq[k]{i},\vec{h}]) \defi \max\{\lh(g),\lh(\vec{h})\} + 1$, and
  \item $\lh(\SNRN[g,h_{w},\vec{s_{w}},\vec{t_{w}}\,(w \in \Sigma^k_0)]) \defi
    \max\{\lh(g),\lh(h_w),\lh(\vec{s_{w}}),\lh(\vec{t_{w}}) \mid w \in \Sigma^k_0\} + 1$.
  \end{enumerate}
  Further, define the safe mapping $\safe$ as indicated by the system $\RN$.
  Then it can be shown that $R_f \subseteq {\epostar}$ for $\epostar$
  induced by $\sp$.
  We only consider the most interesting case, the 
  orientation of the final rule.
  For brevity, we only consider two level of nestings. 
  The argument can be easily extended to the general case.
  Abbreviate $\SNRN[g,h_{w},\vec{s_{w}},\vec{t_{w}}\,(w \in \Sigma^k_0)]$ as $\m{f}$.
  We show 
  \begin{multline*}
    u \defi \m{f}(\mS_{i_1}(;z_1), \dots, \mS_{i_k}(;z_k), \vec{x};\vec{y})  \epostar \\
    h_{i_1 \cdots i_k} (\vec{v_1},\vec{x};\vec{y},
    \m{f}(\vec{v_1},\vec{x};\vec{t_w}(\vec{v_2},\vec{x};\vec{y}, \m{f}(\vec{v_2},\vec{x};\vec{y})))).
  \end{multline*}

  By Definition \ref{d:epostar}.\ref{d:epostar:1}, we obtain
  $u \epostar y_i$ for $y_i \in \vec{y}$. Further 
  Definition \ref{d:subepo} gives $u \subepostar x_i$ for $x_i \in \vec{x}$.
  Thus by Definition \ref{d:epostar}.\ref{d:epostar:3} and the 
  observation below the system $\RN$
  we conclude  
  $u \epostar \m{f}(\vec{v_2},\vec{x};\vec{y})$.
  In particular, the observations on $\vec{v_2}$ also give $u \subepostar v_j$ for $v_1,\dots,v_k = \vec{v_2}$.
  By Definition \ref{d:epostar}.\ref{d:epostar:2} we see
  $
  u \epostar t_w(\vec{v_2},\vec{x};\vec{y}, \m{f}(\vec{v_2},\vec{x};\vec{y}))
  $, 
  by  Definition \ref{d:epostar}.\ref{d:epostar:3} we obtain
  $$
  u \epostar \m{f}(\vec{v_1},\vec{x};\vec{t_w}(\vec{v_2},\vec{x};\vec{y}, \m{f}(\vec{v_2},\vec{x};\vec{y}))) \tpkt
  $$ 
  We conclude with a final application of Definition \ref{d:epostar}.\ref{d:epostar:2}.
\end{proof}


\section{Implementation}\label{s:impl}

We reduce the problem of finding an instance $\epostar$ such that $\RS \subseteq {\epostar}$ holds 
to the \emph{Boolean satisfiability problem} $\SAT$. 
To simplify the presentation, we extend language of propositional 
logic with truth-constants $\top$ and $\bot$ in the obvious way.
To encode the (admissible) precedence $\qp$, we introduce for $f,g \in \DS$ 
propositional variables $\esp{f}{g}$ and $\eep{f}{g}$ to encode the strict and equivalence part
of $\qp$.
We use the standard approach \cite{SFTGACMZ:07} to assert that those variables encode a 
quasi-precedence on $\DS$. Recall that constructors are minimal in the precedence. 
To simplify notation we set
for $f \not\in \DS$ or $g\not\in\DS$
\begin{equation*}
  \esp{f}{g} \defi 
  \begin{cases}
    \top & \text{if $f\in\DS$ and $g\in \CS$,} \\
    \bot & \text{otherwise.}
  \end{cases}
  \quad
  \eep{f}{g} \defi 
  \begin{cases}
    \top & \text{if $f\in\CS$ and $g\in \CS$, } \\
    \bot & \text{otherwise.}
  \end{cases}
\end{equation*}

Further, to encode whether $i \in \safe(f)$ 
we use the variables $\esafe{f}{i}$ for 
$i \in \{1,\dots,n \}$ and $n$-ary $f \in \DS$.
Recall that arguments positions of constructors are always safe. 
We set $\esafe{f}{i} \defi \top$ for $n$-ary $f \in \CS$ and $i \in \{1,\dots,n\}$. 
To increase the strength of our implementation, 
we orient the system $\mu(\RS)$ obtained from $\RS$ by permuting arguments according
to a fixed permutation per function symbol, expressed by mappings
$\ofdom{\mu_f}{\{1,\dots,n\} \to \{1,\dots,n\}}$ for $n$-ary $f \in \FS$. 
The mapping is lifted to terms in the obvious way:
\begin{equation*}
  \mu(t) \defi 
  \begin{cases}
    t & \text{if $t \in \VS$} \\
    f(t_{t_{\mu_f(1)}}, \dots, t_{\mu_f(n)}) & \text{if $t = f(\seq{t})$}.
  \end{cases}
\end{equation*}
We set $\mu(\RS) \defi \{ \mu(l) \to \mu(r) \mid {l \to r} \in \RS \}$.
It is easy to see that $\mu$ does not change derivation heights, in particular, 
$\rcRi[\RS] = \rcRi[\mu(\RS)]$.
To encode the mapping $\mu_f$ for $n$-ary $f \in \FS$ 
we use propositional variables $\emu{f}{i}{k}$ for $i,k \in \{1,\dots,n\}$. 
The meaning of $\emu{f}{i}{k}$ is that argument position
$i$ of $f$ should be considered as argument position $k$, i.e., $\mu(i) = k$, 
compare also \cite{SFTGACMZ:07}. 
We require that those variables encode a permutation on argument positions, which 
is straight forward to formulate in propositional logic. 

To ensure a consistent use of safe argument positions in the constraints below, 
we require that if $f \ep g$, then their arities match and further, 
safe argument positions coincide as expressed 
by the constraint
\begin{align}
  \label{e:impl:precconstraint}
  \compatible \defi \bigand_{f,g \in \FS} \eep{f}{g} \imp 
  \bigand_{i=1}^n \bigand_{j=1}^n \bigand_{k=1}^n \emu{f}{i}{k} \wedge \emu{g}{j}{k} \imp (\esafe{f}{i} \iff \esafe{g}{j}) \tpkt
\end{align}
Here $n$ denotes the arity of $f$ and $g$.

Let $s, t \in \TERMS$ be two concrete terms.
We encode $s \eqi t$ (respecting the argument permutation $\mu$) as 
the constraint $\enc{s \eqi t}$ defined as follows:
\begin{align*}
  \enc{s \eqi t} \defi 
  \begin{cases}
    \top & \text{if $s = t$,} \\
    \eep{f}{g} \wedge {\bigand_{i=1}^{n} \bigand_{j=1}^{n} \bigand_{k=1}^{n}
    \emu{f}{i}{k} \wedge \emu{f}{j}{k} \imp \enc{s_i \eqi t_j}} & \text{if ($\star$),} \\
    \bot & \text{otherwise.}
  \end{cases}
\end{align*}
Here ($\star$) denotes $s=f(\seq{s})$ and $t=g(\seq{t})$.
The comparison $s \subepostar t$ is expressed by
\begin{align*}
  \enc{f(\seq{s}) \subepostar t} \defi 
    \bigor_{i=1}^n c_i \wedge (\enc{s_i \subepostar t} \vee \enc{s_i \eqi t})
\end{align*}
where $c_i = \top$ if $f \in \CS$ and $c_i = \neg \esafe{f}{i}$ if $f \in \DS$.
For $s \in \VS$  we set $\enc{s \subepostar t} \defi \bot$. 
Next we consider the comparison $s \epostar t$, and set 
$$
\enc{s \epostar t} \defi \enc{s \epostarc{1} t} \vee \enc{s \epostarc{2,3} t} \tpkt
$$
Here $\enc{s \epostarc{1} t}$ is the encoding of Case \ref{d:epostar:1}, 
$\enc{s \epostarc{2,3} t}$ expresses Case \ref{d:epostar:2} and Case \ref{d:epostar:3}
from Definition~\ref{d:epostar}.
The constraint $\enc{s \epostarc{1} t}$ is expressed similar to above:
\begin{align*}
  \enc{f(\seq{s}) \epostarc{1} t} \defi 
    \bigor_{i=1}^n \enc{s_i \epostar t} \vee \enc{s_i \eqi t}\tkom
\end{align*}
and $\enc{s \subepostar t} \defi \bot$ for $s \in \VS$.

Let $s = f(\pseq[l][m]{s})$, $t = g(\pseq[k][n]{t})$, and reconsider 
Definition~\ref{d:epostar}.\ref{d:epostar:2}
and Definition~\ref{d:epostar}.\ref{d:epostar:3}.
In both cases we require $s \epostar t_j$ for safe argument positions $j \in \{k+1,\dots,k+n\}$. 
If $f \sp g$, additionally $s \subepostar t_j$ has to hold for all normal argument positions
$j \in \{1,\dots,k\}$. On the other hand, if $f \ep g$, then we need to check 
the \emph{stronger} statement
(i) $s_1 \eqis t_1, \dots, s_{i-1} \eqis t_{i-1}$, 
(ii) $s_i \subepostar t_i$
and (iii) $s \subepostar t_{i+1}, \dots, s \subepostar t_{k}$
for some $i \in \{1, \dots, \min(l,k)\}$. 
Note here that (i) and (ii) (and trivially (iii)) imply $s \subepostar t_i$. 
We encode conditions (i) and (ii) in the constraint 
$\enc{s \epostarlexk[1] t}$
defined below. Then Definition~\ref{d:epostar}.\ref{d:epostar:2}
and Definition~\ref{d:epostar}.\ref{d:epostar:3} is expressible by the constraint
\begin{multline*}
  \enc{f(\seq{s}) \epostarc{2,3} g(\seq[m]{t})} \defi 
  \bigl(\esp{f}{g} \vee \eep{f}{g} \wedge \enc{s \epostarlexk[1] t} \bigr) \\ 
  \wedge \bigand_{j=1}^m (\esafe{g}{j} \imp \enc{s \epostar t_j}) \wedge (\neg \esafe{g}{j} \imp \enc{s \subepostar t_j}) \tpkt
\end{multline*}
For the remaining cases, we set $\enc{s \epostarc{2,3} t} \defi \bot$. 
Further, we set for $k \in \{1,\dots,n \}$
\begin{multline*}
  \enc{s \epostarlexk[k] t} \defi 
  \bigand_{i=1}^{n} \bigand_{j=1}^{n} \bigand_{k=1}^{n} (\emu{f}{i}{k} \wedge \emu{g}{j}{k}
  \imp (\esafe{f}{i} \imp \enc{s \epostarlexk[k+1] t}) \\
  \wedge (\neg \esafe{f}{j} \to (\enc{s_i \subepo t_j} \vee (\enc{s_i \eqis t_j} \wedge \enc{s\epostarlexk[k+1] t})))
\end{multline*}
and $\enc{s \epostarlexk[k] t} \defi \bot$ for $k > n$.

Finally, compatibility of the TRS $\mu(\RS)$ is expressible as the constraint
\begin{align*}
  \epoconstraint \defi \compatible \wedge \precedence \wedge \bijection \wedge \bigand_{{l \to r} \in \RS} \enc{l \epostar r} \tkom
\end{align*}
where $\compatible$ is as defined in Equation \ref{e:impl:precconstraint}, 
$\precedence$ asserts a correct encoding of the admissible quasi-precedence $\qp$
and $\bijection$ asserts that $\mu_f$ for $f \in \FS$ indeed correspond to bijections on argument positions.

\begin{proposition}
  Let $\RS$ be a TRS such that the constraint $\epoconstraint$ is satisfiable. 
  Then $\mu(\RS) \subseteq {\epostar}$ for some argument permutation $\mu$ and exponential path order
$\epostar$. 
\end{proposition}


\section{Conclusion}\label{s:conclusion}

In this paper we  present the \emph{exponential path order} \EPOSTAR.  
Suppose a term rewrite system $\RS$ is  compatible  with \EPOSTAR,  
then  the  runtime  complexity of $\RS$ is bounded from  above by an exponential  function. 
Further, \EPOSTAR\ is sound and complete for the class of functions computable in
exponential time on  a Turing machine. We have implemented
\EPOSTAR\ in the complexity tool \tct.%
\footnote{See \url{http://cl-informatik.uibk.ac.at/software/tct/}, the experimental
data for our implementation is available 
here:~\url{http://cl-informatik.uibk.ac.at/software/tct/experiments/epostar}.}
\tct\ can automatically prove exponential runtime complexity of
our motivating example $\RSfib$. Due to Theorem~\ref{t:soundness} we thus
obtain through an automatic analysis that the computation of the
Fibonacci number is exponential.


\bibliographystyle{plain}

\end{document}